\documentclass[letterpaper, 10 pt, conference]{ieeeconf}
\IEEEoverridecommandlockouts  
\usepackage{graphics} 
\usepackage{epsfig} 
\usepackage{mathptmx} 
\usepackage{amsmath} 
\usepackage{amssymb} 
\allowdisplaybreaks 
\usepackage{cite}
\usepackage{relsize} 
\usepackage[dvipsnames]{xcolor}
\usepackage{graphicx}
\usepackage{amsfonts}
\usepackage{mathtools}
\usepackage[mathcal]{euscript}
\usepackage[hidelinks]{hyperref}
\usepackage[T1]{fontenc} 
\usepackage{mathtools, cuted} 
\usepackage[ruled, vlined, linesnumbered, noend]{algorithm2e}
\let\labelindent\relax 
\usepackage{enumitem} 
\usepackage{balance} 
\usepackage{mathrsfs} 
\usepackage{fdsymbol} 
\usepackage{soul} 
\usepackage{subcaption}

\usepackage[framed,numbered,autolinebreaks,useliterate]{mcode}

\DeclareCaptionLabelSeparator{periodspace}{.\quad}
\let\proof\relax

\usepackage{amsthm}

\UseRawInputEncoding 

\def\bbb{\mathbb B}

\def\bbr{\mathbb R}

\def\bbp{\mathbb P}
\def\bbe{\mathbb E}

\def\bft{\mathbf T}
\def\cala{\mathcal A}
\def\calb{\mathcal B}

\def\calm{\mathcal M}

\def\calp{\mathcal P}

\def\cals{\mathcal S}
\def\calt{\mathcal T}
\def\calu{\mathcal U}
\def\calx{\mathcal X}
\def\caly{\mathcal Y}
\def\calz{\mathcal Z}
\def\scra{\mathscr A}

\def\scru{\mathscr U}
\def\scrx{\mathscr X}

\DeclareMathOperator*{\argmax}{arg\,max}

\newcommand{\p}[1]{{\color{purple}{#1}}}

\newcommand{\absval}[1]{\mid {#1} \mid}

\newcommand{\norm}[1]{\left\lVert {#1} \right\rVert}
\newcommand{\rlbrack}[1]{\left [ {#1} \right]}
\newcommand{\rlbrace}[1]{\left \{ {#1}  \right\}}
\newcommand{\rlpar}[1]{\left ( {#1} \right)}

\newcommand{\matg}{\succ}

\newcommand{\matleq}{\preceq}

\newcommand{\mean}[1]{\overline{{#1}}}
\newcommand{\expect}[1]{\bbe \rlbrack{{#1}}}

\newcommand{\untilI}[1]{U_{{#1}}} 

\newtheorem{remark}{Remark}

\newtheorem{theorem}{Theorem}

\newtheorem{assump}{Assumption}
\theoremstyle{definition}
\newtheorem{example}{Example}
\newtheorem{problem}{Problem}
\newtheorem{definition}{Definition}

\title{\LARGE \bf
Risk-Bounded Temporal Logic Control \\of Continuous-Time Stochastic Systems}
\author{Sleiman Safaoui, Lars Lindemann, Iman Shames, Tyler H. Summers
\thanks{This work was supported in part by the United States Air Force Office of Scientific Research under Grant FA2386-19-1-4073 and in part by the National Science Foundation under Grant ECCS-2047040.}
\thanks{S. Safaoui and T. Summers are with the School of Engineering at The University of Texas at Dallas, Richardson, TX, USA. E-mail: \{sleiman.safaoui, tyler.summers\}@utdallas.edu.
L. Lindemann is with the School of Engineering at the University of Pennsylvania, Philadelphia, Pennsylvania. Email: larsl@seas.upenn.edu.
I. Shames is with the School of Engineering at The Australian National University, Acton, Australia. E-mail: iman.shames@anu.edu.au.
}
}

\begin{document}

\maketitle
\thispagestyle{empty} 

\begin{abstract}
Motivated by the recent interest in risk-aware control, we study a continuous-time control synthesis problem to bound the risk that a stochastic linear system violates a given specification. We use risk signal temporal logic as a specification formalism in which distributionally robust risk predicates are considered and equipped with the usual Boolean and temporal operators. Our control approach relies on reformulating these risk predicates as deterministic predicates over mean and covariance states of the system. We then obtain a timed sequence of sets of mean and covariance states from the timed automata representation of the specification. To avoid an explosion in the number of automata states, we propose heuristics to find candidate sequences effectively. To execute and check dynamic feasibility of these sequences, we present a sampled-data control technique based on time discretization and constraint tightening that allows to perform timed transitions  while satisfying the continuous-time constraints. 
\end{abstract}


\section{Introduction} \label{sec:introduction}
The design of safe control laws for autonomous systems has been studied extensively over the past years. For deterministic systems, the safe control synthesis problem is usually cast as a set invariance problem. Proposed solutions consider control barrier functions \cite{ames2016control}, Hamilton Jacobi reachability analysis \cite{herbert2017fastrack}, or model predictive control \cite{chen1998quasi}. However, when the system is stochastic, e.g., due to uncertainty in the system localization, set invariance has to be interpreted by taking risk into account. Besides safety, the system is subject to performance objectives. As system specification complexity plays a major role in the tractability of the control problem, often simple navigation specifications \cite{dimarogonas2006feedback}, i.e., going from $A$ to $B$ while avoiding obstacles, or regulation and reference tracking problems are studied. This excludes a large class of specifications such as repetitive specifications (always repeating a certain sequence of events), specifications with strict temporal requirements (reaching some state within a specific time interval then reaching another state), and many others. More complex system specifications have recently been considered using spatio-temporal logics \cite{raman2014model,sadraddini2015robust}. In this paper, we hence cast the safe control synthesis problem of stochastic system as a risk-aware control synthesis problem with the goal to upper bound the risk that spatio-temporal logic system specifications are violated.

\textbf{Literature review. } Signal temporal logic (STL) is a real-time temporal logic that allows to impose large classes of specifications \cite{maler2004monitoring}. Importantly, such specifications permit to define robust semantics that provide information as to what extent a specification is satisfied or violated \cite{donze2,fainekos2009robustness}. The deterministic control synthesis problem has been addressed using optimization techniques \cite{raman2014model,sadraddini2015robust,mehdipour2019arithmetic}, machine learning techniques \cite{puranic2021learning,varnai2020robustness,liu2021model}, and automata-based techniques in conjunction with transient control laws \cite{gundana2021event,lindemann2020efficient}. More recently, stochastic control treating the STL specification as a chance constraint have been considered \cite{farahani2018shrinking,sadigh2016safe,jagtap2020formal}. These works consider  specific notions of risk and assumptions on the state distribution such as Gaussian distributions, and they largely study discrete-time systems. On the other hand, risk-aware control for more simple system specification, i.e., not not complex  temporal logic specifications, have been considered in various directions, see e.g., \cite{singh2018framework,chapman2019risk,9147792,schuurmans2020learning}. 

In contrast, in this paper we consider risk-aware control under STL specifications. Most closely related to this paper are our previous works \cite{lindemann2021reactive,safaoui2020control} in which we consider risk-aware control for STL specifications. While the focus in \cite{lindemann2021reactive} is on stochastic environments and reactivity, \cite{safaoui2020control} considers stochastic systems and risk, however in a setting where time is discretized and only for the limited fragment of bounded STL specifications. In this paper, we consider a continuous-time stochastic linear system which, to our knowledge, has not been solved with proper formal guarantees. We restrict our attention to linear systems as the distribution's statistics for a nonlinear system are in general hard to estimate.

\textbf{Contributions. }We continue along the lines of our previous work \cite{lindemann2020efficient} where the continuous-time control synthesis problem for a \emph{deterministic system} under STL specifications is studied. The continuous-time control synthesis problem for a \emph{stochastic system} is more challenging and the state explosion problem can not be addressed as in \cite{lindemann2020efficient}. First, the problem of finding control laws that achieve timed transitions in the mean and covariance states is difficult. This problem becomes even more difficult when the state distribution is not known, as is the case in this paper. Second, the efficient integration of control laws into the automata representation of the specification is non-trivial. Our contributions are:
\begin{itemize}
    \item We present, to the best of our knowledge, the first risk-bounded solution to the continuous-time control synthesis problem of stochastic (non-Gaussian) linear systems under STL specifications.
    \item We propose a sampled-data control technique that performs timed transitions within the space of distributions and guarantees continuous-time constraint satisfaction.
    \item As opposed to existing mixed integer linear programming solutions (presented in the deterministic system literature), our method can handle unbounded STL formulas.
\end{itemize}
\section{Background}

Let $\mathbb{R}$, $\mathbb{R}_{\ge 0}$, and $\mathbb{Q}_{\ge 0}$ denote the sets of real numbers, non-negative real numbers, and non-negative rational numbers respectively. 
Let $\otimes$ denote the Kronecker product.
Given a matrix, the $vec(\cdot)$ operator stacks its vectors. 
$diag(\cdot)$ is a diagonal matrix of the arguments.
$I_n$ and $0_n$ are the $n \times n$ identity and zero matrices. 
Let $(\Omega, \mathcal{F}, \mathbb{P})$ be a probability space where $\Omega$ is the sample space, $\mathcal{F}$ is a $\sigma$-algebra of subsets of $\Omega$, and $\mathbb{P}$ is a probability measure on $\mathcal{F}$. Given random vectors $X_1,X_2: \Omega \rightarrow \bbr^n$, the expected value of $X_1$ with respect to $\bbp$ is denoted by $\expect{X_1}$ and their covariance is $Cov(X_1,X_2) = \bbe[X_1 X_2^T] - \expect{X_1}\expect{X_2}$.
We abbreviate positive semi-definite as psd.

\subsection{Real-Time Temporal Logics}

Signal interval temporal logic (SITL) is a specification formalism that allows describing a desired system behavior for deterministic systems. A predicate $\mu: \bbr^n \to \bbb$ is a Boolean-valued function that depends on a  function $\alpha: \bbr^n \to \bbr$, also referred to as the predicate function. For a given $x \in \bbr^n$, the predicate is true $\mu(x) = \top$ if $\alpha(x) \geq 0$ and false $\mu(x) = \bot$ if $\alpha(x) < 0$. Let $M$ be a set of atomic predicates $M := \rlbrace{\mu_1, \dots, \mu_{\absval{M}}}$. For $\mu\in M$, the SITL syntax is given by
\begin{align}
    \phi \; ::= \; \top \; | \; \mu \; | \; \neg \phi \; | \; \phi_1 \wedge \phi_2 \; | \; \phi_1 \untilI{I} \phi_2 \; 
\end{align}
where $\phi$, $\phi_1$, and $\phi_2$ are STL formulas and where $\untilI{I}$ is the until operator with time interval $I\subseteq \mathbb{Q}_{\ge 0}$ that is not a singleton; $\neg$ and $\wedge$ encode negations and conjunctions. Based on these operators, one can further derive the operators: 
$\phi_1 \vee \phi_2:=\neg(\neg\phi_1 \wedge \neg\phi_2)$ (disjunction operator),
$F_I\phi:=\top \untilI{I} \phi$ (eventually operator), and
$G_I\phi:=\neg F_I\neg \phi$ (always operator).

An SITL formula $\phi$ is evaluated over deterministic signals $x:\mathbb{R}_{\ge 0}\to\mathbb{R}^n$, potential trajectories of a deterministic system. When $x$ satisfies the SITL formula $\phi$ at time $t$, we denote this by $(x,t)\models \phi$. The continuous-time STL semantics \cite{donze2} (define when $x$ satisfies $\phi$ at $t$) are inductively defined as:
\begin{itemize}
    \item $(x,t)\models \mu$ iff $\alpha(x(t))\ge 0$
    \item $(x,t)\models \neg \phi$ iff $(x,t)\not\models \phi$
    \item $(x,t)\models \phi_1\wedge\phi_2$ iff $(x,t)\models \phi_1$ and $(x,t)\models \phi_2$
    \item $(x,t)\models \phi_1 \untilI{I} \phi_2$ iff $\exists t'' \in t\oplus I$, $(x,t'')\models \phi_2$ and $\forall t' \in (t,t'')$, $(x,t')\models \phi_1$
\end{itemize}
An SITL formula $\phi$ is satisfiable if $\exists x$ such that $(x,0)\models\phi$. 
Such an SITL formula can be translated into a timed automaton $TST_\phi$ \cite{lindemann2020efficient} (see Appendix \ref{app:SITL} for a brief summary). From a timed automaton $TST_\phi$, one can obtain plans which can be thought of as requirements on how each predicate $\mu_i(x(t))$ in $M$ and hence the signal $x(t)$ has to evolve over time $t$ (see \cite{lindemann2020efficient} for details).  While SITL is defined over deterministic signals (and hence for deterministic systems), in this paper, we are interested in stochastic systems.

\subsection{Continuous Time Stochastic System}
We consider the stochastic linear system described by a stochastic differential equation
\begin{align}
    dX(t) = (AX(t) + Bu(X(t)))dt + dW(t), \quad X(0) = X_0, \label{eq:ct_stoch_lin_sys}
\end{align}
where $A \in \bbr^{n \times n},\ B \in \bbr^{n \times m}$ are the (constant) dynamics and input matrices.
We assume that $(A,B)$ is stabilizable.
$X(t) \in \bbr^n$ is the state, $u:\bbr^n\to \bbr^m$ is a state feedback control law, and $X_0:\Omega \to \bbr^n$ is a random variable describing the unknown initial state. 
$dW(t)$ is differential Brownian motion. The stochastic integral $\int_{0}^{T}dW = W(T)-W(0)$ is a possibly non-Gaussian random variable with zero mean $\mean{W} := \expect{W(t)} = 0 \ \forall t$, and covariance $\Sigma := \bbe[W(t) W(t)^T] = diag(\sigma_1^2, \dots, \sigma_n^2)$ where $\sigma_i^2 := Var(W_i(t)) = Cov(W_i(t), W_i(t))$.
Often, $W(t)$ is assumed to be Gaussian (some continuous-time stochastic systems literature assumes that without explicitly stating it).
However, in this work we \emph{do not} make such assumptions. Instead, to promote robustness to uncertainties in the distribution, we consider a moment based ambiguity set:
\begin{align}
    \calp^W := \{ \bbp^W | & \bbe_{\bbp^W}[W] = \mean{W},  \label{eq:W_ambig_set} \\
    & Cov_{\bbp^W}(W) = \bbe_{\bbp^W}[(W-\mean{W})(W-\mean{W})^T] = \Sigma \} \nonumber
\end{align}
that is, $W(t)$ can belong to any distribution with mean $\mean{W}$ (in this case $\mean{W} = 0_n$) and psd covariance matrix $\Sigma$ both of which are assumed to be known.

In this paper, we consider the feedback control law 
\begin{align}
    u(X(t)) = K(t) X(t) + k(t), \label{eq:feedback_ctrl}
\end{align}
which we motivate in the next section,
where $K(t) \in \bbr^{m \times n}$ is a feedback gain and $k(t) \in \bbr^m$ is an open-loop control signal.
The dynamics in \eqref{eq:ct_stoch_lin_sys} thus become:
\begin{align}\label{eq:feedback-ct_stoch_lin_sys}
    d X(t) = (A+B K(t)) X(t)dt + Bk(t)dt + dW(t). 
\end{align}

\subsection{Mean and Covariance Dynamics}
Consider the dynamics in \eqref{eq:feedback-ct_stoch_lin_sys}. The state mean is denoted by $\mean{X}(t):=\bbe[X(t)]$ and the mean dynamics are given by:
\begingroup
        \makeatletter\def\f@size{10}\check@mathfonts
\begin{align}
    \dot{\mean{X}}(t) &:= \bbe[\dot{X}(t)] = (A+B K(t)) \mean{X}(t) + Bk(t), \ \mean{X}(0) =: x_0 \label{eq:mean_dyn}
\end{align}
\endgroup
The state covariance is denoted by $P(t) := Cov(X(t), X(t))$ and the covariance dynamics are given by 
\begin{align}
    & \dot{P}(t) = (A+BK(t))P(t) + P(t)(A+BK(t))^T + \Sigma, \ P(0) = P_0 \nonumber \\
    & \iff 
    vec(\dot{P}(t)) = (I_n \otimes (A+BK(t)) + \nonumber \\
    & \qquad \qquad \qquad (A+BK(t)) \otimes I_n)vec(P(t)) + vec(\Sigma). \label{eq:covar_dyn}
\end{align}

\begin{remark}
    The choice of feedback control law \eqref{eq:feedback_ctrl} results in having $k$ in the mean dynamics \eqref{eq:mean_dyn} (feedforward term) and $K$ in the covariance dynamics \eqref{eq:covar_dyn} (limit covariance growth). 
\end{remark}

The mean and covariance dynamics are both given by first order ODEs and can be stacked into a single ODE:
\begin{align}
    \dot{\calx}(t) &= 
    \cala(K(t))\calx(t) + \calb\calu(t) \label{eq:aug-dyn}
\end{align}
where
\begin{align*}
    \cala(K) &:= 
    \begin{bmatrix}
        A+BK & 0 \\ 
        0 & (I_n \otimes (A+BK) + (A+BK) \otimes I_n)
    \end{bmatrix} \\
    \calb &:= 
    \begin{bmatrix}
        B & 0 \\
        0 & I_{n\cdot n}
    \end{bmatrix}, \quad 
    \calx(t) := 
    \begin{bmatrix}
        \mean{X}(t) \\ vec(P)
    \end{bmatrix} \\
    \calu(t) &:= 
    \begin{bmatrix}
        k(t) \\ vec(\Sigma)
    \end{bmatrix} \quad \text{($vec(\Sigma)$ is a constant).}
\end{align*}
We chose to include $vec(\Sigma)$ (constant) in $\calu$ to retain the traditional linear system format. We also assume the following.
\begin{assump} \label{assump:dynamics_assumptions}
    We assume that there exist closed sets $\scrx,\scru$ and $\calm \in \bbr$ such that $\norm{\dot{\calx}} \leq \calm \ \forall (\calx, \calu) \in \scrx \times \scru$.
\end{assump}
Since \eqref{eq:aug-dyn} is linear, it is  Lipschitz continuous. This, and Assumption \ref{assump:dynamics_assumptions} will become relevant in \S\ref{sec:timed_trans}. 
Regarding the set $\scrx$, 1) we consider bounds on the mean state (e.g. physical bounds on the system) so that $\mean{X}_{min} \matleq \mean{X} \matleq \mean{X}_{max}$ where $\mean{X}_{min}, \mean{X}_{max} \in \bbr^n$ are known, and 2) we assume $A+BK(t)$ is stable by the choice of $K(t)$, thus since $\Sigma$ is psd, $P$ will have a unique symmetric psd steady state value which is the solution to a Lyapunov equation \cite[Thm 22]{simon2006optimal}; since the covariance dynamics are linear, $vec(P)$ will be bounded $0 \matleq vec(P) \matleq D \in \bbr^{n\cdot n}$ for some $D$. 
As for $\calu$, $vec(\Sigma)$ is a constant and we assume that the open-loop control $k(t)$ is bounded $k_{min} \matleq k \matleq k_{max}$.

\section{Problem Statement}

In the stochastic framework introduced above, the satisfiability problem of  an SITL formula $\phi$ is ill-posed, i.e., whether or not a stochastic linear system $X(t)$ satisfies the formula $\phi$. In fact, an atomic predicate $\mu(X(t))$ becomes a random variable. We hence redefine the atomic predicates as risk predicates. Intuitively, the truth value of a risk predicate $\mu^{Ri}(X)$ is true if the risk of violating the predicate is small. 

The risk of violation is obtained by using a risk measure, i.e. a function that maps a random variable to a real number. Let $C$ denote all measurable functions from the sample space $\Omega$ to $\mathbb{R}^n$, i.e., all random variables. Then, a risk measure is defined as $\rho: C \rightarrow \bbr$. We use the distributionally robust value at risk (DR-VaR) which is a coherent risk metric defined as $\inf_{\mathbb{P} \in \calp}\mathbb{P}[-\alpha(X) \leq 0] \geq 1-\eta$ for some risk threshold $\eta\in(0,1)$. The DR-VaR satisfies certain desirable axioms \cite{zymler2013worst}, see \cite[\S II]{safaoui2020control} for more details on risk measures.

Formally, we define the risk predicate as:
\begin{align}\label{eq:risk_predicate}
    \mu^{Ri}(X)&:=\begin{cases}
    \top & \text{if } \rho(-\alpha(X)) \le \eta \\
    \bot &\text{otherwise }
    \end{cases}
\end{align} 
The risk SITL (RiSITL) syntax is defined as
\begin{align}\label{eq:RiSITL}
    \phi \; ::= \; \top \; | \; \mu^{Ri} \; | \; \neg \phi \; | \; \phi_1 \wedge \phi_2 \; | \; \phi_1 \untilI{I} \phi_2. \; 
\end{align}
where $\mu^{Ri}\in M^{Ri}$ for the risk predicates $M^{Ri}=\{\mu_1^{Ri},\hdots,\mu_{|M|}^{Ri}\}$, while the other operators have the same meaning as in SITL. 

The semantics of RiSITL are different in how the risk predicates are evaluated. Instead of $(x,t)\models \mu$ iff $\alpha(x(t))\geq 0$ in SITL, we have $(X,t)\models \mu^{Ri}$ iff $\rho(-\alpha(X)) \le \eta$, while the other operators follow as in the SITL semantics presented earlier. 
$(X,t)\models \phi$ indicates that the stochastic linear system with dynamics $\dot{X}$ satisfies the RiSITL formula $\phi$ at time $t$.
We are now ready to state the formal problem definition.
\begin{problem}\label{prob:1}
    Given the system \eqref{eq:feedback-ct_stoch_lin_sys} and an RiSITL $\phi$ per \eqref{eq:RiSITL}, find the control variables $K(t)$ and $k(t)$ so that $(X,0)\models \phi$.
\end{problem}

\section{Risk-Bounded Temporal Logic Control}

Our solution to Problem \ref{prob:1} consists of a reformulation of the risk predicates (\S\ref{sec:AP_reform}), an optimization-based controller (\S\ref{sec:timed_trans}) for timed transitions (see Definition \ref{def:timed_transition} later), and the decomposition of the RiSITL specification into a sequence of timed transitions (\S\ref{sec:planning} - \S\ref{sec:risk_based_ctrl_special}). 
In particular, we generate a candidate sequence of timed automaton transitions that we feed into our optimization-based controller to check for dynamic feasibility of this candidate sequence. We note that our solution is sufficient, i.e.,  sound but not \emph{complete}.

\subsection{Atomic Predicate Reformulation} \label{sec:AP_reform}
Consider now linear predicate functions $\alpha(X) = a^T X + b$ where $a \in \bbr^n, \ b \in \bbr$. The following reformulation holds 
\begin{align}
    \rho(-\alpha(X)) \leq \eta 
    \Leftrightarrow & \rho(-(a^TX+b)) \leq \eta \nonumber\\
    \Leftrightarrow & \inf_{\mathbb{P} \in \calp}\mathbb{P}[-a^T X -b \leq 0] \geq 1-\eta  \nonumber\\
    \Leftrightarrow & a^T \mean{X} + b - \underbrace{\sqrt{\frac{1-\eta}{\eta}}}_{:=H} \norm{P^{1/2}a}_2 \geq 0 \label{eq:tightened_const_drvar}
\end{align}
where the last step follows from \cite[Thm 3.1]{calafiore2006distributionally}. 
The risk predicate \eqref{eq:risk_predicate} is now a \emph{deterministic} risk-tightened predicate:
\begin{align}\label{eq:reform_risk_predicate}
    \mu^{Ri}(X)&:=\begin{cases}
    \top & \text{if } a^T \mean{X} + b - H \norm{P^{1/2}a}_2 \geq 0 \\
    \bot &\text{otherwise. }
    \end{cases}
\end{align}
\begin{remark} \label{rmk:risk_tightened_halfspace}
    One way to interpret \eqref{eq:reform_risk_predicate} is to view $\mu^{Ri}(X)$ as a time varying halfspace in the mean dynamics $\mean{X}$. This follows as $P(t)$ is specified apriori when $K(t)$ is chosen in advance. Intuitively, the halfspaces will be tightened so that predicate functions $\alpha(X)$ specifying ``goal regions'' shrink and ``obstacle regions'' expand with higher uncertainty.
\end{remark}

\subsection{Control of Timed Transitions}
\label{sec:timed_trans}
Consider now two polytopes $S_1$ and $S_2$ that are subsets of $\mathbb{R}^{n+n\cdot n}$ and that are either connected or intersecting. We define the timed transition problem as follows.
\begin{definition}[Timed Transition]\label{def:timed_transition}
    Given polytopes $S_1$ and $S_2$ and a transition time $T>0$, then a feedback gain matrix $K(t)$ and a control law $\calu(t)$ achieve a timed transition from $\calx(0)\in S_1$ into $S_2$ at time $T$ if the following is satisfied:
    \begin{subequations}\label{eq:prob_timed_trans}
        \begin{align}
            & \dot{\calx}(t) = \cala(K(t))\calx(t) + \calb\calu(t) \quad &\forall t \in [0,T] \label{eq:prob_timed_trans_a}\\
            & \calu(t) \in \scru, \quad &\forall t \in [0,T] \label{eq:prob_timed_trans_b} \\
            & \calx(t) \in S_1, \quad &\forall t \in [0,T) \label{eq:prob_timed_trans_c} \\
            & \calx(T) \in S_2.\label{eq:prob_timed_trans_d}&
        \end{align}
    \end{subequations}
\end{definition}
A timed transition under a feedback gain matrix $K(t)$ and a control law $\calu(t)$ hence occurs when the mean and the covariance  dynamics as well as the inputs bounds  $\scru$ are respected, and the system state $\calx(t)$ transitions from the set $S_1$, in which the system starts, into the set $S_2$ at time $T$. 

In \S\ref{sec:planning}, the sets $S_1$ and $S_2$ will encode conjunctions of predicates $\mu^{Ri}(X)$ as per \eqref{eq:reform_risk_predicate}. For instance, we may have $S_1:=\{(\mean{X},vec(P))\in\mathbb{R}^{n+n\cdot n} \mid a^T \mean{X} + b - H \norm{P^{1/2}a}_2 \geq 0 \}$ for a single predicate function. It is clear that the polytopes are convex in $\mean{X}$. For example, $S_1$ may encode a room and $S_2$ may encode a corridor next to the room. For a robot with initial state in $S_1$, a timed transition here would require the robot to stay in the room and transition into the corridor exactly at time $T$. Towards finding a feedback gain matrix $K(t)$ and a control law $\calu(t)$ we propose the following architecture.

\subsubsection{Finding $K(t)$} \label{sec:finding_K}
To efficiently integrate our control law with the timed automata representation of $\phi$ (presented in the next section), we propose to first select the closed-loop feedback gain $K(t)$ for a timed transition. Doing so will allow us to determine the covariance matrix $P(t)$ at all times. Without loss of generality, we consider a constant gain $K$ over the entire timed-transition period and require $K$ to satisfy three requirements (related to Assumption \ref{assump:dynamics_assumptions}):
\begin{itemize}
    \item $K$ must be stabilizing, i.e. $A+BK$ must be stable
    \item $K$ must keep the dynamics $\dot{\calx}$ bounded by $\calm$ 
    \item $K$ must keep the covariance bounded by $D$.
\end{itemize}
Additionally, $K$ should be chosen to keep $\calm, D$ small enough to avoid the problem becoming infeasible when the atomic predicates \eqref{eq:reform_risk_predicate} are tightened.
We note that this choice of $K$ may not exist even if a solution to \eqref{eq:prob_timed_trans} exists. This is one factor that makes our solution methodology only sufficient.

\subsubsection{Finding $\calu(t)$} \label{sec:finding_k}
Once $K$ is determined, the covariance matrix $P(t)$ becomes fully defined for all times $t\ge 0$. We now find a control law $k(t)$, which determines $\calu(t)$, to achieve the timed transition as in \eqref{eq:prob_timed_trans} by solving the optimization problem where open-loop control $k(t)$ is minimized:
\begin{subequations}\label{eq:prob_timed_trans_}
\begin{align}
    \min_{k(t)} \quad & \int_0^T k(t)^T R k(t) \\
    \text{s.t.} \quad &\;\eqref{eq:prob_timed_trans_a}-\eqref{eq:prob_timed_trans_d}.
\end{align}
\end{subequations}
where $R$ is psd. Note that \eqref{eq:prob_timed_trans_} is convex in $k(t)$, but infinite dimensional due to the continuous time formulation resulting in an infinite number of decision variables and constraints. Solving this finite horizon, continuous-time, constrained, open-loop control problem exactly may be possible with Pontryagin's optimality principle, but is difficult in general. However, to facilitate a sampled-data implementation, we solve the above optimization problem by discretizing time as in \cite{fontes2018guaranteed} while tightening the constraints to be able to satisfy the original constraints in continuous time.

Our Assumption \ref{assump:dynamics_assumptions} is similar to \cite[Asm 1]{fontes2018guaranteed}.
We note that $\calm$ can be found by solving the following problem:
\begin{align*}
    \calm = \argmax \quad & \norm{\dot{\calx}} \\
    \text{s.t.} \quad & \calx \in \scrx, \quad \calu, \in \scru.
\end{align*}

Let us next define the signed distance function $dist_S: \bbr^{n+n\cdot n} \rightarrow \bbr$ similar to \cite[(5)]{fontes2018guaranteed} as:
$dist_S(\calx) := \min_{\caly \in S} \norm{\calx- \caly} - \inf_{\calz \in S^C} \norm{\calx - \calz}$
i.e., if $\calx \in S$, then the distance is the negative distance to the boundary of $\calx$ while if $\calx \not\in S$ the distance is the positive distance to the boundary. Hence, constraints of the form $\calx(t) \in S$ are equivalent to $dist_S(\calx(t)) \leq 0$.
For simplicity, we opt for a uniform discretization of the time interval $[0,T]$ as $\bft := \rlbrace{\calt_0, \dots \calt_N}$  such that $\calt_{0} = 0$, $\calt_{N} = T$, and $\delta \calt = \frac{T}{N}$ is the time step. Then, using \cite{fontes2018guaranteed} we can guarantee continuous-time satisfaction of the constraints in \eqref{eq:prob_timed_trans_a}-\eqref{eq:prob_timed_trans_d} by tightening them further. 
With $K(t) = K$ for all $t \in [0,T]$ selected as discussed before and assuming $k(t)$ is a zero-order hold control law, so that we have $k[\calt]$ for $\calt \in \bft$, we propose the following problem to find a controller that satisfies \eqref{eq:prob_timed_trans} in continuous-time:
\begin{subequations}\label{eq:timed_transition_reformulation}
\begin{align}
    \min_k \quad & \sum_{i \in \{0, \dots N\}} k[\calt_{i}]^T R k[\calt_{i}] \\
    \text{s.t.} \quad & \dot{\calx}(t) = \cala(K)\calx(t) + \calb\calu(t) \quad \forall t \in [0, T] \label{eq:ct_dynamics_reform} \\
    & \calu(\calt) \in \scru, \quad \calt \in \bft\\
    & dist_{S_1}(\calx(\calt)) \leq -\calm \cdot \delta \calt, \quad \calt \in \bft \label{eq:eps_in_S1} \\
    & \calx(T) \in S_2. \label{eq:S2_satisfaction}
\end{align}
\end{subequations}

\begin{theorem} \label{thm:1}
    Given $K$ and $\calm$ as in \S\ref{sec:finding_K}. Let $k[\calt]$ be a solution to \eqref{eq:timed_transition_reformulation}. 
    Then the constraints of the timed transition problem \eqref{eq:prob_timed_trans_} are satisfied,
    i.e. a timed transition from $S_1$ into $S_2$ as per Definition \ref{def:timed_transition} can be achieved under $K$ and $k[\calt]$.
\end{theorem}
\proof{The proof follows by applying \cite[Theorem 1]{fontes2018guaranteed} with the augmented dynamics \eqref{eq:aug-dyn}.}

\begin{remark}
    The dynamics \eqref{eq:ct_dynamics_reform} are deterministic linear dynamics in $\calx$ with a closed form solution at the sampled times given by: $\calx(\calt) = e^{\cala(K)\calt}\calx(0) + \int_{0}^{\calt}(e^{\cala(K)(\calt-t)}Bk(t))dt$ where the interval can be split into several integrals because $k(t) = k[\calt_i]$ for $t \in [\calt_i, \calt_{i+1})$ \cite{simon2006optimal}.
\end{remark}

\begin{figure}[!htb]
    \centering 
    \includegraphics[width=0.45\textwidth, trim={3cm 2cm 3cm 4cm},clip]{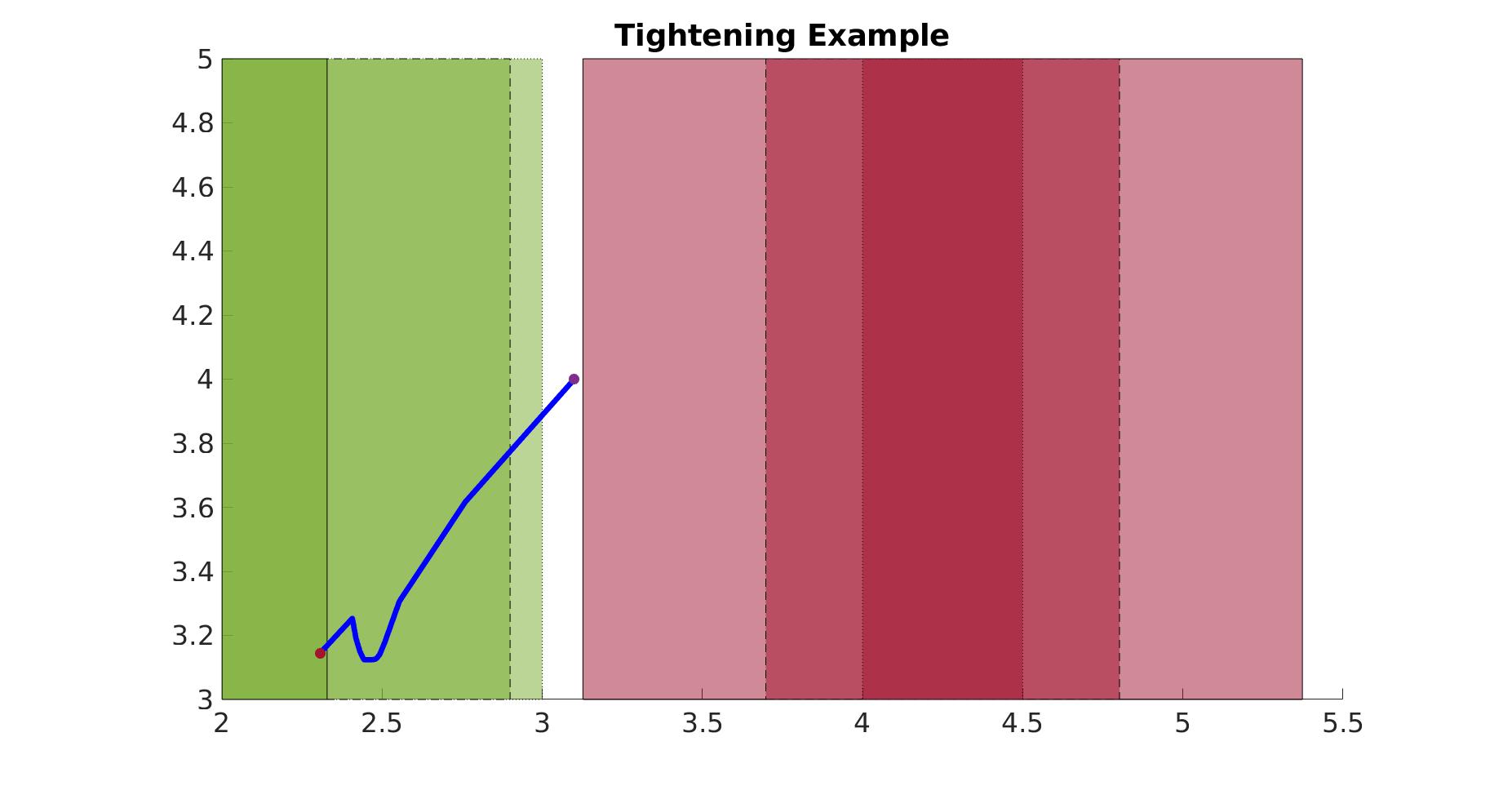}
    \caption{Constraints: not tightened (dotted), DR-tightening (dashed), and DR- and DT-tightening (solid) for a goal region (left, green) and obstacle region (right, red). A timed transition (blue) moves the robot into the goal.
    }
    \label{fig:tightening_example}
    \vspace{-0.3cm}
\end{figure}

\begin{example}
    To demonstrate the different tightening procedures described, consider three atomic predicate functions in 2D 
    $\alpha_1 := [-1, 0]X + 3 \geq 0$,
    $\alpha_2 := [-1, 0]X + 4 \geq 0$, and 
    $\alpha_3 := [1, 0]X -4.5 \geq 0$.
    The first defines a goal region (reach) and the other two define an obstacle region (avoid).
    We use a risk bound $\eta = 0.5$ for $\mu_1^{Ri}(X)$ and $\eta = 0.1$ for $\mu_2^{Ri}(X)$ and $\mu_3^{Ri}(X)$ which are the risk-tightened predicates as in \eqref{eq:reform_risk_predicate}.
    For single integrator dynamics with $K(t) = diag(-5, -5) \ \forall t$ and $\delta \calt = 0.01$, we have $\calm \cdot \delta \calt = 0.5712$.
    With $\Sigma = diag(0.1, 0.1)$, $P(0) = 0_2$, and $T=2$, we get $P(2)= diag(0.0101, 0.0101)$. 
    In Fig. \ref{fig:tightening_example}, we plot the (not tightened) constraints, the risk-tightened (DR-tightened) constraints at $t=2$ and the risk and discrete-time-tightened (DR- and DT-tightened) constraints at $t=2$ with dotted, dashed, and solid edges respectively.
    Notice how the goal ``shrinks'' while the obstacle ``expands'' as $P(t)$ increases.
    The timed transition (blue) is the solution to \eqref{eq:timed_transition_reformulation} that moves the robot into the goal (from the purple to the red point). 
    The discrete-time points stay between the DR- and DT-tightened predicates of the goal and obstacle for all $t \in \{\calt_0, \dots, \calt_{N-1}\}$, thus satisfying \eqref{eq:eps_in_S1}. At $t=T$, the agent satisfies the DR-tightened predicates \eqref{eq:S2_satisfaction}.
\end{example}

\subsection{Risk-Based Automaton} \label{sec:planning}

As mentioned before, avoiding the state explosion problem for stochastic systems is challenging and can not be addressed as in \cite{lindemann2020efficient}. This is particularly the case as the efficient integration of control laws for stochastic systems into the timed-automata representation of the specification is non-trivial.
In our proposed problem solution, we first translate the RiSITL specification $\phi$ into a timed automaton $TST_\phi$ similarly to \cite{lindemann2020efficient}, but with risk predicates $\mu^{Ri}$ instead of predicates $\mu$. The procedure closely follows Appendix \ref{app:SITL}, but we will next briefly describe the difference for RiSITL and refer to the appendix for more intuition. We abstract $\phi(M^{Ri})$ into a MITL specification $\varphi(P)$ (see Appendix \ref{app:MITL} for an introduction to MITL) where $P$ is a set of propositions that replaces the set of risk predicates $M^{Ri}$. The MITL formula $\varphi$ is then translated into a timed automaton $TST_\varphi$ according to \cite{ferrere2019real}. 
Each state $s$ and transition $\delta$ in $TST_\varphi$ now encode intersections of constraints of the form
$\{(\mean{X},vec(P))\in\mathbb{R}^{n+n\cdot n}\mid a^T \mean{X} + b - H \norm{P^{1/2}a}_2 \geq 0 \}$ which we denote by $\lambda(s)$ and $\lambda(\delta)$. We then perform the following operations on $TST_\varphi$ to obtain $TST_\phi$:
\begin{enumerate}
    \item[[$\overline{O1}$]] 
    Remove any $s$ in $TST_\varphi$ if $\not\exists(\mean{X},P)$, $(\mean{X},P) \models \lambda(s)$.
    \item[[$\overline{O2}$]] 
    Remove each transition $\delta$ in $TST_\varphi$ if $\not\exists(\mean{X},P)$ such that $(\mean{X},P) \models \lambda(\delta)$.
\end{enumerate}
Operations $\overline{O1}$ and $\overline{O2}$ are taken without the consideration of the system dynamics \eqref{eq:aug-dyn}. With respect to Remark \ref{rmk:risk_tightened_halfspace}, we view $P(t)$ as apriori given when $K(t)$ is chosen in advance so that the constraints $\lambda(s)$ and $\lambda(\delta)$ can be seen as time-varying in the mean state. We can hence, instead, perform operations $\overline{O1}$ and $\overline{O2}$ for a fixed 
$P$ and obtain a tightened automaton
denoted by 
$TST_\phi^\text{tight}$. Some choices of $P$ are:
\begin{itemize}
    \item $P = 0_n$: $[\overline{O1}]$, $[\overline{O2}]$ become $[O1]$, $[O2]$ from \cite{lindemann2020efficient}
    \item $P \in \{P_{max,i} \mid \forall i \in \{1, \dots |M|\} \}$ such that $P_{max, i} = \arg \sup_P \norm{P^{1/2}a_i}_2$ where $0 \matleq vec(P) \matleq D$: 
    Here, every predicate \eqref{eq:reform_risk_predicate} is tightened by the maximum amount of its $\norm{P^{1/2}a}_2$ term (we call this \emph{maximum tightening}).
\end{itemize}
In these cases, [$\overline{O1}$] and [$\overline{O2}$] can be constructed as simple feasibility problems as described in \cite{bemporad1999control}.

\begin{figure}[!tb]
    \centering 
    \includegraphics[width=0.45\textwidth]{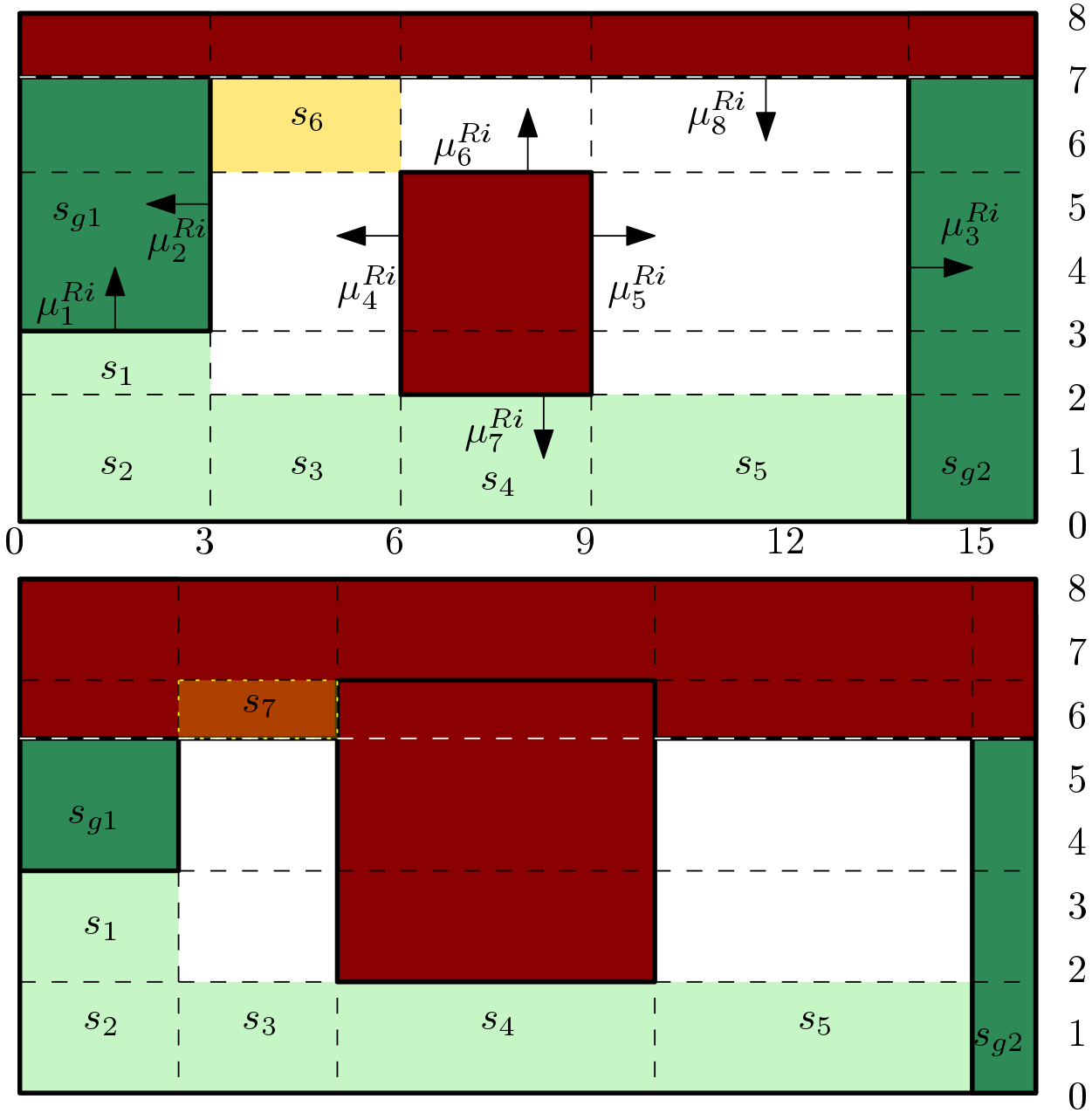}
    \caption{Automaton states changing as $P$ changes}
    \label{fig:automaton_states}
    \vspace{-0.3cm}
\end{figure}

\begin{example}\label{ex:automaton_states}
    Consider the illustrative example in Fig. \ref{fig:automaton_states}. 
    Ignoring the environment bounds, the goal and obstacle regions are defined with eight atomic predicates $\mu^{Ri}_1, \dots, \mu^{Ri}_8$. 
    The atomic predicates are indicated with dashed lines and the halfspace where the predicate is satisfied is indicated by an arrow (top figure).
    There are two dark green goal regions: 1) left $g_1 := \mu^{Ri}_1 \wedge \mu^{Ri}_2$ and 2) right $g_2:=\mu^{Ri}_3$. Their atomic predicates point inwards.
    There are two red obstacle regions: 1) middle $o_1:=\neg \mu^{Ri}_4 \wedge \neg \mu^{Ri}_5 \wedge \neg \mu^{Ri}_6 \wedge \neg \mu^{Ri}_7$ and 2) top $o_2:=\neg \mu^{Ri}_8$. Their atomic predicates point outwards.
    The top figure is for $P(0)=0_2$. The bottom is for some stead state $P(\infty) \matg P(0)$ (its value is irrelevant here).
    Some examples of automaton states are:
    $s_1 := \neg \mu^{Ri}_1 \wedge \mu^{Ri}_2 \wedge \neg \mu^{Ri}_3 \wedge \mu^{Ri}_4 \wedge \neg \mu^{Ri}_5 \wedge \neg \mu^{Ri}_6 \wedge \neg \mu^{Ri}_7 \wedge \mu^{Ri}_8$ and 
    $s_7 := \mu^{Ri}_1 \wedge \neg \mu^{Ri}_2 \wedge \neg \mu^{Ri}_3 \wedge \mu^{Ri}_4 \wedge \neg \mu^{Ri}_5 \wedge \mu^{Ri}_6 \wedge \neg \mu^{Ri}_7 \wedge \neg \mu^{Ri}_8$.
    Notice that with $P(0)$ (top), $s_6$ is non-empty, but with $P(\infty)$ (bottom) it is empty and removed by $[\overline{O1}]$. 
    The reverse happens for $s_7$.
    The states $s_1, s_2, s_3, s_4, s_5$ are non-empty and remain connected for all tightening values $P(0) \matleq P(t) \matleq P(\infty)$.
\end{example}

\subsection{Risk-Based Control: General Case} \label{sec:risk_based_ctrl_general}
Towards an efficient solution, we propose to find candidate sequences of timed transitions from $TST_\phi^\text{tight}$ and post-hoc check for dynamic feasibility by means of the control technique derived in \S\ref{sec:timed_trans}.  The tightened automaton $TST_\phi^\text{tight}$ helps us to guide the search process to find a feasible solution and to decrease the search space.

One approach is to consider the maximum tightening case described earlier. This is the most conservative approach and results in the most robust solution.
[$\overline{O1}$] and [$\overline{O2}$] are applied to $TST_\phi$ with every atomic predicate maximally tightened (e.g. Fig. \ref{fig:automaton_states} (bottom)). The resulting automaton is $TST_\phi^\text{tight}$.

From the tightened automaton $TST_\phi^\text{tight}$, we find a sequence of timed automaton transitions using graph search techniques (see \cite{lindemann2020efficient,alur1996benefits} for details). These automaton transitions are defined by alternating transitions between automaton states (in zero time) and transitions within the same automaton state (in finite time). This returns a sequence of the form, 
$(s_0, 0) \xrightarrow[]{\delta_0} 
(s_1, 0) \xrightarrow[]{\tau_1}
(s_1, \tau_1) \xrightarrow[]{\delta_2}
(s_2, \tau_1) \xrightarrow[]{\tau_2}
(s_2, \tau_1+\tau_2) \xrightarrow[]{\delta_3}
(s_3, \tau_1+\tau_2)$.
$(s_0, 0) \xrightarrow[]{\delta_0} (s_1, 0)$ is an example of a transition between automaton states,
and $(s_1, 0) \xrightarrow[]{\tau_1} (s_1, \tau_1)$ is a transition within the same state.
The former transition occurs instantaneously; it is the exact moment the switch happens from one automaton state to another (e.g. the exact moment a robot reaches the boundary between two rooms). 
The latter transition occurs in the same automaton state while the dynamics evolve in time (e.g. the robot moves in one room to reach the boundary to the next room).
To simplify notation we use: 
$(s_1, 0) \xrightarrow[]{\tau_1, \delta_2} (s_2, \tau_1) :=
(s_1, 0) \xrightarrow[]{\tau_1}
(s_1, \tau_1) \xrightarrow[]{\delta_2}
(s_2, \tau_1)$.

The timed transitions of Definition \ref{def:timed_transition} accomplish $(s_i, \tau_i) \xrightarrow[]{\tau, \delta} (s_j, \tau_i+\tau)$.
For example, consider $(s_1, 0) \xrightarrow[]{\tau_1, \delta_2} (s_2, \tau_1)$.
The automaton state $s_1$ corresponds to a convex set of mean states that satisfy its input label $\lambda(s_1)$, i.e. $S_1:=\lambda(s_1)$. This and $\tau_1$ being the transition time $T$ are handled by \eqref{eq:eps_in_S1}. Similarly, $s_2$ corresponds to the mean states of $S_2:=\lambda(s_2)$ (hence \eqref{eq:S2_satisfaction}).

\begin{algorithm}
    \SetAlgoLined
    \SetKwInOut{Input}{input}
    \SetKwInOut{Output}{output}
    \Input{Candidate automaton transition sequence}
    \Output{$\calu(t)$ or Failure}
    Choose $K(t) = K$ for this transition\; \label{line:find_K}
    $s = s_0$, $t = 0$\; \label{line:automaton_start}
    \For{$\tau, \delta$ in automaton transition sequence such that $(s, t) \xrightarrow[]{\tau, \delta} (s', t+\tau)$ }{
        solved, $k(\calt)$ = solve\_timed\_transition($s$, $\tau$, $\delta$)\; \label{line:find_k}
        \uIf{not solved}{
            break\;
        }
        Store $k(t)$\;
        $s = s_{next}$, $t = t + \tau$\;
    }
    \caption{Risk-Based Timed Transition Control}  \label{algo:algo}
\end{algorithm}

Then, for a timed automaton transitions sequence from $TST_\phi^\text{tight}$, we use Algorithm \ref{algo:algo} to check if it is dynamically feasible. 
For every $(s_i, \tau_i) \xrightarrow[]{\tau, \delta} (s_j, \tau_i+\tau)$, we use \S\ref{sec:finding_K} and \S\ref{sec:finding_k} to find $K$ and $k(t)$ (lines \ref{line:find_K}, \ref{line:find_k}). If we fail to find them, we stop and try a different sequence of timed automaton transitions. If all transitions succeed, then we have found $u(t)$ and Problem \ref{prob:1} is solved.
Note that this heuristic approach is \emph{only sufficient} to finding a solution.

\begin{theorem}
    Consider the automaton $TST_\phi^\text{tight}$ obtained by performing operations $[\overline{O1}]$ and $[\overline{O2}]$ for some fixed matrix $P$, e.g., the maximal tightening. If a sequence of timed automaton transitions $(s_i, \tau_i) \xrightarrow[]{\tau, \delta} (s_j, \tau_i+\tau)$ can be found such that the corresponding timed transitions (Definition \ref{def:timed_transition}) can be achieved by $K$ and $k(t)$ as per \eqref{eq:timed_transition_reformulation}, then the linear stochastic system $X(t)$ satisfies the RiSITL formula $\phi$, i.e. $(X, 0) \models \phi$ and Problem \ref{prob:1} is solved.
\end{theorem}
\begin{proof}
    Ignoring the dynamics, the existence of a sequence of timed automaton transitions ensures that the formula $\phi$ is satisfiable \cite[Lemma 2]{lindemann2020efficient}. From Theorem \ref{thm:1}, if $K$ and $k(t)$ are found for an automaton transition, then the resulting trajectory is dynamically feasible and satisfies the continuous-time risk-based constraints. Applying this to every automaton transition completes the proof.
\end{proof}

\begin{remark} \label{rem:satisfaction_of_almost_all_timesteps}
    In practice, when solving \eqref{eq:timed_transition_reformulation} repeatedly in Algorithm \ref{algo:algo}, the DT-tightening may sometimes render \eqref{eq:timed_transition_reformulation} infeasible.
    A simple workaround, which works well for small $\delta \calt$, is to remove this DT-tightening term $\calm \cdot \delta \calt$ from \eqref{eq:eps_in_S1} for a few initial and final sampled times (e.g. $\calt_0, \calt_1, \calt_2$ and $\calt_{N-2}, \calt_{N-1}, \calt_N$).
\end{remark}

\subsection{Risk-Based Control: Special Case} \label{sec:risk_based_ctrl_special}
We consider a special case which allows for a less-conservative solution.
While the method can be generalized to $\bbr^n$, we will discuss the $\bbr^2$ case only. 
Consider $a \in \{[0\ \pm1]^T, [\pm1 \ 0]^T \}$. Thus, an automaton state $s$ is a 2D rectangular area.
We thus have $a^T P(t) a \in \{ \pm P(t)_{(1,1)}, \pm P(t)_{(2,2)} \}$ (i.e. $\pm$ the first or second diagonal elements of $P(t)$) and hence $\norm{P^{1/2}a}_2 = \sqrt{a^T P a} \in \{ \sqrt{P_{(1,1)}}, \sqrt{P_{(2,2)}} \}$. 
This makes it easier to describe the conditions on $P(t)$ for which an automaton state $s$ is non-empty (i.e. for which there exists $(\mean{X}, P)$ such that $(\mean{X}, P) \models \lambda(s)$). 
Example \ref{ex:simple_example} illustrates these conditions.

\begin{figure}[!htb]
    \centering
    \includegraphics[width=0.3\textwidth]{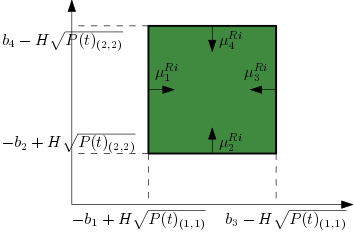}
    \caption{Special case $P$ conditions example}
    \label{fig:simple_example}
    \vspace{-0.5cm}
\end{figure}
\begin{example} \label{ex:simple_example}
Consider $\mu^{Ri}_1 \wedge \mu^{Ri}_2 \wedge \mu^{Ri}_3 \wedge \mu^{Ri}_4$ in Fig. \ref{fig:simple_example}.
For the rectangular set to be non-empty, the distances between parallel edges must be positive (note $P_{(1,1)}, P_{(2,2) \geq 0}$, $P$ is psd):
\begin{itemize}
    \item $-b_1 + H \sqrt{P_{(1,1)}} \leq b_3 - H \sqrt{P_{(1,1)}}$
    $\iff \rlpar{\frac{b_1+b_3}{2}}^2 \geq P_{(1,1)}$
    \item $-b_2 + H \sqrt{P_{(2,2)}} \leq b_4 - H \sqrt{P_{(2,2)}}$
    $\iff \rlpar{\frac{b_2+b_4}{2}}^2 \geq P_{(2,2)}$.
\end{itemize}
\end{example}

We now introduce operation [$\overline{O3}$] to augment the states of the automaton $TST_{\phi}$ with conditions on $P(t)$. 
\begin{enumerate}
    \item[[$\overline{O3}$]] For all $s$ in $TST_{\phi}$, add the two conditions on $P(t)_{(1,1)}$ and $P(t)_{(2,2)}$ (a la Example \ref{ex:simple_example}) for which $s$ is non-empty.
\end{enumerate}
After [$\overline{O3}$], the resulting $TST_{\phi}$ has the same states but with extra guards on the states to guide the planning search. 
[$\overline{O3}$] can be applied to $TST_{\phi}$ after [$O1$], [$O2$]. Then, we can follow the same procedure described in \S\ref{sec:risk_based_ctrl_general}.

\section{Numerical example}
\begin{figure}
    \centering 
    \includegraphics[width=0.48\textwidth, trim={8cm 2cm 7cm 4cm},clip]{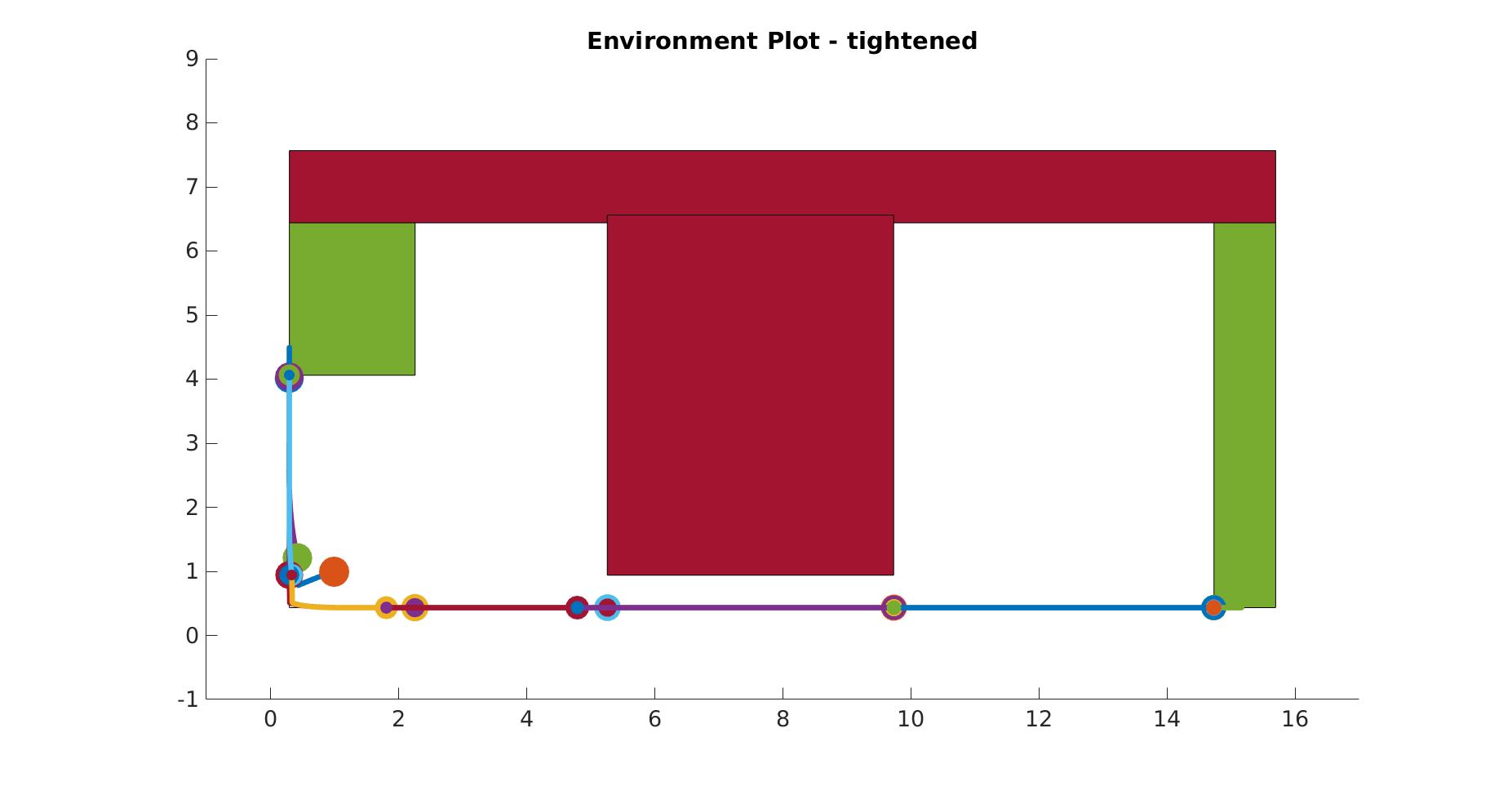} \\
    \includegraphics[width=0.48\textwidth, trim={8cm 2cm 7cm 4cm},clip]{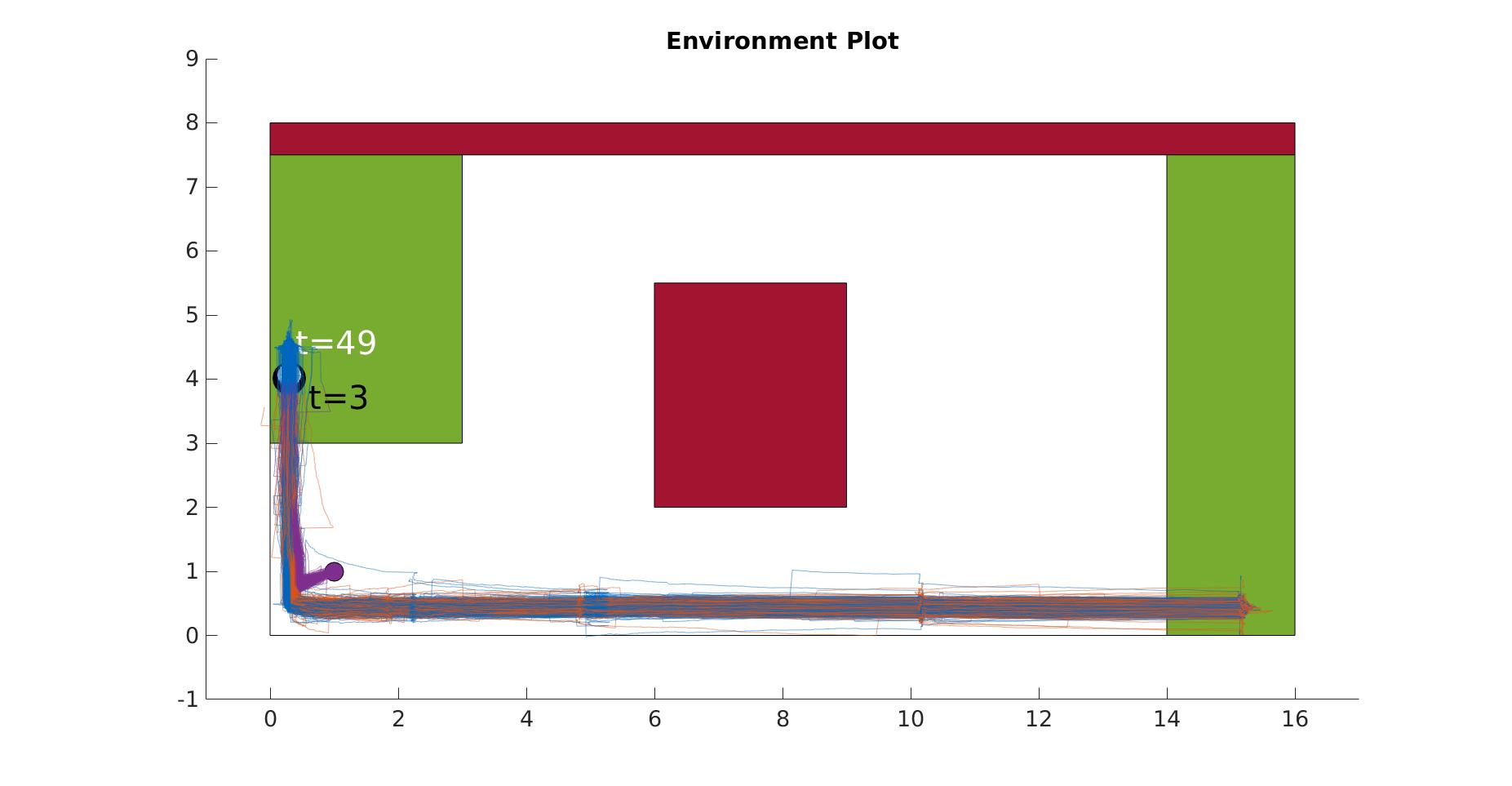}
    \caption{(Top) Timed transitions of a solution for the considered example. Circles mark the start and end of transition; their radius decreases as time increases. (Bottom) 10000 Monte Carlo rollouts for the given solution. Trajectories: (purple) from the initial condition to left goal, (blue) from the left goal the to right goal, (orange) from the right goal to the left goal.}
    \label{fig:sol_tightening}
    \vspace{-0.3cm}
\end{figure}

Consider the environment of Figure \ref{fig:automaton_states} (top) and a robot with dynamics $A = diag(0.07, 0.1)$, $B = diag(1,1)$. 
We use $k_{max} = -k_{min} = 30$, $\delta \calt = 0.01$, $\Sigma=diag(0.1, 0.1)$, $K = diag(-0.9, -0.5)$ for all transitions, and risk bound $\eta=0.1$ for all constraints except the environment bounds where we use $\eta_{env}=0.4$. 
For the specification, consider the repetitive task
$\phi = G_{[0, 3+23 \omega]}(\neg o_1 \wedge \neg o_2) \wedge F_{[0,3]} g_1 \wedge G_{[3,3+23 \omega]}(F_{[0,11]} g_2 \wedge F_{[11,23]}g_1)$ where $\omega$ is the number of repetitions of going from $g_1$ to $g_2$ and back 
($g_1, g_2, o_1, o_2$ were defined in Example \ref{ex:automaton_states}).

Using the approach in \S\ref{sec:risk_based_ctrl_special}, with the robot starting at $x_0 = [1,\ 1]$ with $P(0) = 0_2$ we get the sequence of timed automaton transitions:
$ 
(s_2, 0) \xrightarrow[]{\tau_1=1, \delta_1}
(s_1, 1) \xrightarrow[]{\tau_2=2, \delta_2}
(s_{g1}, 3) [\xrightarrow[]{\tau_3=1, \delta_3}
(s_1, 4) \xrightarrow[]{\tau_4=2, \delta_4}
(s_2, 6) \xrightarrow[]{\tau_5=2, \delta_5}
(s_3, 8) \xrightarrow[]{\tau_6=2, \delta_6}
(s_4, 10) \xrightarrow[]{\tau_7=2, \delta_7}
(s_5, 12) \xrightarrow[]{\tau_8=2, \delta_8}
(s_{g2}, 14) \xrightarrow[]{\tau_9=2, \delta_9}
(s_5, 16) \xrightarrow[]{\tau_{10}=2, \delta_{10}}
(s_4, 18) \xrightarrow[]{\tau_{11}=2, \delta_{11}}
(s_3,20) \xrightarrow[]{\tau_{12}=2, \delta_{12}}
(s_2, 22) \xrightarrow[]{\tau_{13}=2, \delta_{13}}
(s_1, 24) \xrightarrow[]{\tau_{14}=2, \delta_{14}}
(s_{g1}, 26)]
^{\omega}
$
where $[\cdot]^{\omega}$ represents repeating the transitions $\omega$-times
(the robot starts in $s_2$; all these automaton states are marked in Fig. \ref{fig:automaton_states}).
We use \eqref{eq:timed_transition_reformulation} to find $k(t)$ per Algorithm \ref{algo:algo}.
The resulting sequence of timed transitions, for $\omega=2$, is shown in Fig. \ref{fig:sol_tightening} (top) with the constraints tightened using the steady-state covariance $P(\infty) = diag(0.0602, 0.125)$. The repeated transitions appear superimposed since they are almost identical. 
The trajectory starts at $(1,1)$, moves up to the first goal, then repetitively visits the second goal then the first goal.
We run 10000 Monte Carlo simulations.
The noise is sampled from a 3 degree of freedom, 0-mean, $\Sigma$-covariance, student-t distribution. 
The rollouts are plotted in the original environment in Fig. \ref{fig:sol_tightening} (bottom). The black and white circle in $g_1$ are reached at $t=3$ and $t=49$sec respectively.
The robot never collided with the obstacles, but violated the environment bounds in $0.3\% \ll \eta_{env} = 40\%$ of the cases. 

\section{Conclusion}
We present a risk-bounded controller for  continuous-time stochastic, non-Gaussian, linear system under STL specifications. In particular, we use RiSITL to specify constraints rooted in axiomatic risk theory and reformulate DR-VaR constraints into deterministic risk-tightened constraints.
Then, we consider timed transitions and tighten our constraints further to account for the discrete-time implementation of a sampled data system without loosing continuous-time guarantees.
From there, we show how these timed transitions can be used to verify the dynamically-feasibility of timed automaton transitions from a risk-based automaton.

\appendix

\subsection{SITL to Timed Signal Transducers}
\label{app:SITL}
An SITL formula $\phi$ can be translated into a language equivalent timed signal transducer, i.e., a timed automaton \cite{lindemann2020efficient}. We will need some of the machinery presented in \cite{lindemann2020efficient} despite working with RiSITL, and hence summarize the translation from SITL to timed signal transducer. The first step is to abstract the SITL formula $\phi$ into an MITL formula $\varphi$ (see Appendix \ref{app:MITL} for a description of MITL). 
We use $\phi(M)$ to make explicit that the SITL formula $\phi$ depends on the set of predicates $M$. 
We abstract the SITL formula $\phi(M)$ into an MITL formula $\varphi(P)$  essentially by replacing predicates $M$ in $\phi(M)$ by a set of propositions $P$.  
For $i\in\{1,\hdots,|M|\}$, associate with each $\mu_i\in M$  a proposition $p_i$ and let $P:=\{p_1,\hdots,p_{|M|}\}$. Let then $\varphi(P)=\phi(P)$. 

The translation from MITL to timed signal transducer mainly follows \cite{ferrere2019real}. Let ${c}\in\mathbb{R}_{\ge 0}^O$ be a vector denoting $O$ clocks. These clocks can be reset by the reset function  $r:\mathbb{R}_{\ge 0}^O\to \mathbb{R}_{\ge 0}^O$. Clocks evolve with time when visiting a state of a timed signal transducer, while clocks may be reset during transitions between states. We define clock constraints as Boolean combinations of conditions of the form $c_o\le k$ and $c_o\ge k$ for some $k\in\mathbb{Q}_{\ge 0}$. Let $\Phi({c})$ denote the set of all clock constraints over clock variables in ${c}$. 
\begin{definition}[Timed Signal Transducer \cite{ferrere2019real}]
	A timed signal transducer is a tuple $TST:=(\cals,s_0,\Lambda,\Gamma,{c},\iota,\Delta,\lambda,\gamma, \scra)$ where $\cals$ is a finite set of states, $s_0$ is the initial state with $s_0\cap \cals=\emptyset$, $\Lambda$ and $\Gamma$ are a finite sets of input and output variables, respectively, $\iota:\cals\to\Phi({c})$ assigns clock constraints over ${c}$ to each state, $\Delta$ is a transition relation so that $\delta=(s,g,r,s')\in\Delta$ indicates a transition from $s\in \cals\cup s_0$ to $s'\in \cals$ satisfying the guard constraint $g\subseteq \Phi({c})$ and resetting the clocks according to $r$; $\lambda:\cals\cup\Delta\to BC(\Lambda)$ and $\gamma:\cals\cup\Delta\to BC(\Gamma)$ are input and output labeling functions where $BC(\Lambda)$ and $BC(\Gamma)$ denote the sets of all Boolean combinations over $\Lambda$ and $\Gamma$, respectively, and $\scra\subseteq 2^{\cals\cup \Delta}$ is a generalized B\"uchi acceptance condition.
\end{definition} 

To construct a timed signal transducer $TST_\varphi$ that encodes the MITL formula $\varphi$, we follow the algorithm presented in \cite{ferrere2019real}. For a signal ${d}:\mathbb{R}_{\ge 0}\to\mathbb{B}^{|P|}$, it holds that  $({d},0)\models\varphi$ if and only if $d$ satisfies the generalized B\"uchi acceptance condition of $TST_\varphi$ (see \cite{ferrere2019real} for a definition). As the MITL formula $\varphi$ is an abstraction of the SITL formula $\phi$, we perform two operations on $TST_\varphi$ to obtain
the timed signal transducer 
$TST_\phi$ 
encoding
$\phi$, as presented in \cite{lindemann2020efficient}: 
\begin{enumerate}
    \item[{[O1]}] Remove any $s\in \cals$ if $\not\exists{x}\in\mathbb{R}^n$ so that ${x}\models \lambda(s)$. 
    Remove the corresponding $s$ from $\scra$. 
	\item[{[O2]}] Remove any $\delta:=(s,g,r,s')\in\Delta$ if $\not\exists{x}\in\mathbb{R}^n$ so that $x\models \lambda(\delta)$. Remove the corresponding $\delta$ from $\scra$.
\end{enumerate}

The modified $TST_\varphi$ is denoted by $TST_{\phi}:=(\cals^{\phi},s_{0},\Lambda,\Gamma,{c},\iota,\Delta^{\phi},\lambda,\gamma, \scra^{\phi})$.
For a signal ${x}:\mathbb{R}_{\ge 0}\to\mathbb{R}^n$, it holds that  $({x},0)\models\phi$ if and only if $x$ satisfies the generalized B\"uchi acceptance condition of $TST_\phi$

\subsection{Metric Interval Temporal Logic (MITL)}
\label{app:MITL}
While system specifications are given as SITL formulas, an intermediate step is needed when one wants to obtain a timed signal transducer that encodes the SITL formula. The difference between  SITL and MITL is that SITL considers predicates, while MITL considers propositions. Instead of $\mu(x)$, MITL only considers propositions $p$ where $p = \top$ if the proposition holds and $p = \bot$ if the proposition does not hold. Propositions are hence abstractions of predicates so that MITL can be seen as an abstraction of SITL. Let $P$ be a set of propositions. For $p\in P$, the MITL syntax is: $\varphi \; ::= \; \top \; | \; \p \; | \; \neg \varphi \; | \; \varphi_1 \wedge \varphi_2 \; | \; \varphi_1 \untilI{I} \varphi_2 \; $
with a similar interpretation of the operators as for SITL. Note that we use $\varphi$ to denote MITL formulas, while we use $\phi$ to denote SITL formulas. 

MITL semantics are  similar to SITL semantics. An MITL formula $\varphi$ is interpreted over a Boolean signal $d:\bbr_{\ge 0}\to \bbb^{\absval{P}}$ that corresponds to truth values of the propositions in $P$ over time. Define the projection of $d$ onto $p\in P$ as $proj_p(d):\mathbb{R}_{\ge 0}\to \mathbb{B}$. The only difference of the MITL semantics compared to the SITL semantics is now that instead of $(x,t)\models \mu$ iff $\alpha(x(t))\geq 0$, we have $(d,t)\models p$ iff $proj_p({d})(t)=\top$, while the other operators follows as in the SITL semantics \cite[Sec. 4]{ferrere2019real}. The expression $(d,t)\models \varphi$  indicates that $d$ satisfies the MITL formula $\varphi$ at time $t$. 

\bibliographystyle{IEEEtran}
\bibliography{references}

\end{document}